\documentclass[11pt]{article}

\usepackage{fullpage}
\usepackage{float}
\usepackage{epsfig}
\usepackage{amsmath}
\usepackage{amssymb}
\usepackage{amstext}
\usepackage{amsmath}
\usepackage{xspace}
\usepackage{latexsym}
\usepackage{verbatim}
\usepackage{url}
\usepackage[ruled,vlined]{algorithm2e}

\newtheorem{theorem}{Theorem}[section]
\newtheorem{definition}{Definition}
\newtheorem{claim}{Claim}
\newtheorem{lemma}[theorem]{Lemma}

\newtheorem{example}{Example}

\newcommand{\qed}{\mbox{\ \ \ }\rule{6pt}{7pt} \bigskip}

\renewcommand{\comment}[1]{}
\newenvironment{proof}{\noindent{\em Proof:}}{\hfill\qed}

\newcommand{\argmax}{\operatorname{argmax}}

\newcommand{\val}[1]{v_{#1}}
\newcommand{\bid}[1]{b_{#1}}

\newcommand{\agentval}[2]{#1^{#2}}

\begin{document}
\title{Coopetitive Ad Auctions}

\author{Darrell Hoy\thanks{Northwestern University, Evanston, IL. This work was done while the author was an intern at eBay Research Labs. }
\and Kamal Jain\thanks{eBay Research Labs, San Jose, CA.}
\and Christopher A. Wilkens\thanks{University of California at Berkeley, Berkeley, CA. This work was done while the author was an intern at eBay Research Labs. }
}

\date{}
\maketitle{}

\thispagestyle{empty}

\begin{abstract} A single advertisement often benefits many parties, for example, an ad for a Samsung laptop benefits Microsoft. We study this phenomenon in search advertising auctions and show that standard solutions, including the status quo ignorance of mutual benefit and a benefit-aware Vickrey-Clarke-Groves  mechanism, perform poorly. In contrast, we show that an appropriate first-price auction has nice equilibria in a single-slot ad auction --- all equilibria that satisfy a natural cooperative envy-freeness condition select the welfare-maximizing ad and satisfy an intuitive lower-bound on revenue.

\end{abstract}

\section{Introduction}

In 1991, Intel launched its ``Intel Inside'' advertising campaign and forever changed the way people buy computers. Previously, buyers only considered hardware insofar as it affected the software that would run on their new machine. The ``Intel Inside'' campaign aimed to change that behavior --- Intel coordinated with PC vendors to advertise not just the processor's capabilities but the Intel brand. Twenty years later, the ``Intel Inside'' mark has become one of the most recognized in the tech industry, their signature five-note chime is known worldwide, and, most importantly, buyers think about the brand of processor inside their computers~\cite{II}.

Intel's benefit from the ``Intel Inside'' campaign is an obvious example of a general phenomenon. Intel clearly has a vested interest in the sale of computers containing its products and, in an amortized sense, derives a specific benefit from every sale. This exemplifies a fundamental aspect of marketing: {\em a single advertisement often benefits many different companies.} Companies commonly recognize this benefit and team up with partners in so-called {\em cooperative advertising} agreements similar to the ``Intel Inside'' campaign --- in 2000, an estimated \$15 billion was spent on cooperative advertising in the United States alone~\cite{N06}. This phenomenon is also recognized in the operations research and marketing fields where it has been modeled using a variety of Stackelberg and dynamic games~\cite{B72,HPS09}. However, one key question seems to have gone unasked in both the practical and theoretical realms: {\em how can the companies who sell advertising space exploit the broad benefit of a single ad?} We study this question in the context of online ad auctions.

Our main results show that an auctioneer may improve both his own revenue and consumers' welfare by using an auction that allows cooperation among the advertisers bidding on a single ad but maintains competition between ads --- we call this a coopetitive\footnote{Coopetition is a business term describing an environment where the same parties are simultaneously cooperating in some areas and competing in others~\cite{W}.} ad auction. We first show that conventional cooperative advertising contracts and the Vickrey-Clarke-Groves (VCG) mechanism may perform poorly --- Figure~\ref{fig:silsearch} shows a real query in which Google's current ad auction produces an unreasonable and unsustainable outcome. In contrast, we show that equilibria of the first-price auction which satisfy a cooperative envy-freeness condition have a natural performance guarantee similar to that of a second-price auction.

\begin{figure}[htb]\centering
\includegraphics[scale=0.5]{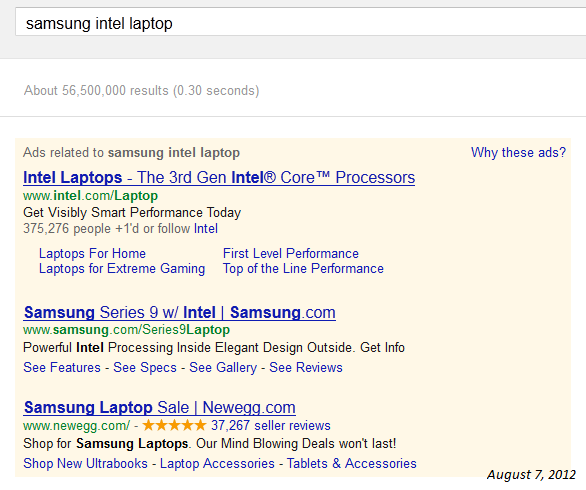}
\caption{A Google search for ``samsung intel laptop'' illustrates the pitfalls of ignoring the mutual benefits of an ad --- the top ad for ``Intel Laptops'' is competing against the ad below it for ``Samsung...w/ Intel.'' Google is charging Intel a premium to have its ad shown on top; however, Intel should be happy if a user clicks on the Samsung ad (and possibly even the Newegg ad) and thus should be unwilling to pay this premium.}\label{fig:silsearch}
\end{figure}

\paragraph{Mutual Benefit and Cooperative Advertising.} Most advertisements (directly or indirectly) benefit many parties. For example, an iPhone ad benefits cell phone providers, a Samsung ad for a Windows laptop benefits Microsoft, and an ad for the Boston Red Sox benefits a bar across the street from Fenway Park. Moreover, these secondary benefits are often significant --- when Best Buy sells a Samsung laptop running Windows, Microsoft makes more money from the laptop's sale than either Best Buy or Samsung. As a result, companies have strong incentives to share advertising costs, particularly when competing against a more integrated adversary as when Microsoft and Samsung compete against Apple.

The status quo technique for pooling advertising dollars is an ad-hoc system of external contracts. In such external contracts, one company agrees to pay a portion of another's advertising costs when its branding is included in the ads. In the Intel example, Intel agrees to pay a percentage of the advertising costs when Samsung or Dell include the ``Intel Inside'' branding in their ad.

In the research community, cooperative advertising has mostly been studied in marketing and operations research. The setting is typically modeled as a Stackelberg game in which an upstream manufacturer makes an offer to downstream manufacturers or retailers~\cite{B72,K98}, sometimes incorporating dynamic components~\cite{HPS09}. Other recent research has relaxed the assumption that the upstream manufacturer is the first mover and considers external contracts based on alternative bargaining solutions~\cite{HLM02}.

\paragraph{Pay-Per-Click Ad Auctions.} Search advertising today is sold through a pay-per-click (PPC) ad auction. In the standard setting, each bidder comes to the auction with its own ad and places a bid in terms of its willingness to pay for each click. The auctioneer subsequently assigns ads to slots on the web page --- bidders effectively compete for slots --- but only charges a bidder when his ad is clicked.

In a standard PPC ad auction, the only way for advertisers to share costs is through external contracts. Unfortunately, this creates undesirable results in an auction format. First, a manufacturer may compete with itself. For example, Figure~\ref{fig:silsearch} shows a Google query for ``samsung intel laptop'' in which an ad from Intel is shown above a Samsung ad that explicitly advertises Intel-based laptops. In general, Intel will be happy if the buyer visits Samsung's site with the intent to buy an Intel-based laptop. Thus, even if Intel would prefer the user to click on its own ad instead of Samsung's, Intel should be unhappy paying the premium required to beat Samsung in the auction.

A second downside to external cooperative advertising agreements is that downstream producers may face a moral hazard. In some situations, the downstream producer may be incentivised to overspend on advertising and waste the upstream manufacturer's money. In the worst case, the upstream manufacturer will refuse to participate and cooperation will collapse entirely. We give examples of both phenomena in Section~\ref{sec:ex}.

The theory of mechanism design suggests another approach: bids should express complex preferences over ads. In particular, it should be possible for multiple bidders to value a single ad and for a single bidder to value multiple ads. Again, however, some natural mechanisms may have poor performance. For example, in Section~\ref{sec:ex} we show that a Vickrey-Clarke-Groves (VCG) auction may generate little or no revenue --- the implicit cooperation among advertisers (which is desired) mimics the kind of strategic collusion that is known to reduce revenue in VCG mechanisms.

\paragraph{The First-Price Coopetitive Ad Auction.} In contrast to the pitfalls of external contracts and VCG mechanisms, our results show that first-price auctions which consider bidders' complex preferences have nice equilibria in single slot settings. First, we generalize the intuition of a second-price auction to give a natural lower-bound on the revenue that the auctioneer should expect. Second, analogous to the results of Edelman et al.~\cite{EOS07} for the commonly used generalized second-price (GSP) auction, we show that all equilibria satisfying a natural cooperative envy-freeness condition maximize welfare while satisfying this revenue lower-bound. We also show that the cooperative envy-free equilibria dominate VCG in terms of revenue --- each bidder is individually paying more in the first-price auction than in the VCG mechanism.

Next, we show that such cooperative envy-free equilibria can be found easily. The envy-freeness constraints define a polytope of which the equilibria form the Pareto frontier, suggesting a family of convex programs for computing equilibria. Finally, we specifically identify the egalitarian equilibrium and give an efficient algorithm for computing it.

First-price auctions are often preferable to VCG and GSP auctions. Most singificantly, they are more transparent --- since prices precisely correspond to bids, bidders do not face uncertainty in their payments and there is no opportunity for the auctioneer to manipulate the auction, particularly by learning in a repeated setting. For example, eBay historically offered a VCG-like auction for selling multiple identical items. The auction was sufficiently disliked by the bidders that eBay no longer offers it as an option. Today, eBay sells single items through a transparent ascending price auction. In the specific case of ad auctions, Overture's experience has discredited the first-price auction; however, this could have been because Overture's auction did not have a pure strategy equilibrium.

\paragraph{Related Work.} In addition to the afore-mentioned work on cooperative advertising and online ad auctions, our work is related to auctions with externalities. Whereas simple auction models assume players are indifferent to the bundle received by another player, in reality there may be externalities, i.e. players may care about the bundles received by other agents. Incorporating externalities has been studied in both the economics and computer science literatures (e.g.~\cite{KMSW10,JMS99}), typically producing a mechanism in which bidders can express a different value for every possible outcome. Our ad auctions may not always have nice interpretations in terms of externalities; however, like mechanisms with externalities they cannot be expressed in the standard bidding language and, in the extreme, degenerate to a mechanism which requires  bidders to express a value for every possible outcome.

\newcommand{\ex}{\mathbf{E}}

\section{The Coopetitive Ad Auction}\label{sec:prelim}

The coopetitive advertising model generalizes the standard pay-per-click (PPC) advertising auction. As in the standard auction, the auctioneer must choose which of $m$ competing ads to show in $s$ slots. Each advertiser derives a value of $v$ from a click on one of its ads and has a utility that is quasilinear in money, i.e. $u(p)=v-p$ when the advertiser gets a click and pays $p$. The likelihood that a user clicks on ad $j$ in slot $k$, called the click-through-rate (CTR), is given by $c_{j,k}$. Hence, the expected utility of bidder whose ad is shown in slot with CTR $c$ is $\ex[u_i(p_i)]=c(v-p)$. In this paper, we focus on the special case with one slot ($s=1$) where CTRs are independent of the ad. In this case, the CTR can be taken to be 1 without loss of generality, so we drop CTRs for the remainder of the paper.

The new feature of the coopetitive model is that an advertiser can derive value from clicks on multiple ads. In general, advertiser $i$'s value for a click $v_{i,j}$ depends on the particular ad $j$; however, for the sake of presentation we consider a simpler model. An advertisement $j\in\{1,\dots,m\}$ is defined by a publicly-known set of advertisers $S_j\subseteq[n]$ who all derive value from a click on advertisement $j$. Advertiser $i$ derives the same value $v_i$ from a click on any ad $j$ where $i\in S_j$ (an advertiser does not benefit if $i\not\in S_j$). We use $T$ to denote the ad set of bidders in the ad with the maximum total value, i.e. $T=\argmax_{S_j}\sum_{i\in S_j}v_i$ (in case of a tie, $T$ denotes the particular winning ad chosen by the auction). All our results generalize to the more complicated $v_{i,j}$ setting by considering equilibria where a bidder bids for the same surplus $v_{i,j}-b_{i,j}$ in each ad.

We will often use shorthand for examples. Our notation is itself best described by an example:
\[\{ (\agentval{A}{2}, \agentval{B}{1}, \agentval{C}{3}), (\agentval{A}{2}, \agentval{B}{1}), (\agentval{C}{3})\}\]
This denotes an auction with three advertisers $A$, $B$, and $C$ whose values are 2, 1, and 3 respectively. Everyone derives value when the first ad $(\agentval{A}{2},\agentval{B}{1},\agentval{C}{3})$ is shown, while only advertisers $A$ and $B$ benefit from the second ad and only advertiser $C$ benefits from the third.

\paragraph{The External Contracts Mechanism.} External contracts are the way advertisers currently cooperate: each ad is ``owned'' by a single bidder, and any party wishing to increase the bid on an ad must negotiate an external contract with the ad's owner.

We model this as a standard VCG PPC ad auction in which advertiser $i$ can (before the auction is run) commit to pay an $\alpha_i$ fraction of the cost each time one of its ads is clicked, up to a maximum $\beta_i$. This payment goes directly to the owner $o(j)$ of the clicked ad $j$. Thus, the utility of advertiser $i$ would be $u_i=v_i-\min(\alpha_i p_{o(i)}, \beta_i)$, while the utility of the ads owner $o(j)$ would be $u_{o(j)}=v_{o(j)}-p_{o(j)}+\min(\alpha_i p_{o(j)},\beta_i)$.

\paragraph{The VCG Mechanism.} In our setting with a single slot, the Vickrey-Clarke-Groves (VCG) mechanism chooses the ad $j$ maximizing $\sum_{i\in S_j}v_i$ and charges bidder $i$ the minimum value he needed to have to be in the winning ad. (Bidders submit their true values $v_i$ because VCG is incentive compatible.)

\paragraph{The First-Price Auction.} In a first-price auction, each advertiser submits a bid $b_i$. The auctioneer displays the ad $j$ maximizing $\sum_{i\in S_j}b_i$ and charges each bidder in the winning ad $p_i=b_i$ when the ad is clicked.

\section{Pitfalls of Standard Mechanisms}\label{sec:ex}

Many standard mechanisms behave poorly with respect to advertisers coopetitive valuations. For example, Figure~\ref{fig:silsearch} shows how Google's current system caused Intel to compete against itself.

\paragraph{External Contracts.} A few pitfalls specifically arise in the current system of external contracts. First, if contracts are made with insufficient granularity, an advertiser might easily compete with itself:
\begin{example}
	Consider the following two single-slot ad auctions:
\[\{(\agentval{S}{3},\agentval{M}{10}), (\agentval{D}{2},\agentval{M}{10}), (\agentval{A}{11})\}\quad\mbox{and}\quad\{(\agentval{S}{3},\agentval{M}{10}), (\agentval{D}{2},\agentval{M}{10})\}\enspace.\]
In the first ad auction, $M$ will happily contribute advertising funds to help beat $A$; however, in the second auction $M$ wins regardless. The only effect of $M$'s dollars in the second auction is to fund a useless bidding war between the $SM$ and $DM$ ads.
	\end{example}
Ideally, $M$ would only contribute advertising funds in the first auction. However, since the granularity of real cooperative advertising contracts is somewhat limited, this example demonstrates a legitimate concern.

Additionally, the advertiser receiving the external contract faces a moral hazard: he is often incentivised to overspend the money of his advertising partner. In equilibrium, the result is that cooperation collapses:
\begin{example}
	Consider the following three ads with four interested parties:
\[\{(\agentval{S}{3},\agentval{M}{10}), (\agentval{D}{2},\agentval{M}{10}), (\agentval{A}{11})\}\enspace.\]
Suppose there are three ad slots with CTR's 0.1, 0.08, and 0.05 respectively (if a bidder appears in multiple ads, their likelihood of a click is the sum of the likelihoods for those ads).

In the external contracts model described in Section~\ref{sec:prelim}, $M$ will not offer a cooperative advertising contract in equilibrium. As a result, the auction degenerates to $\{(\agentval{S}{3}), (\agentval{D}{2}), (\agentval{A}{11})\}$. Not only will revenue be hurt significantly, but the auction will be inefficient because $(\agentval{A}{11})$ wins the top slot.
	\end{example}
We omit the calculations.

\paragraph{The VCG Mechanism.} The VCG mechanism charges a player based on the externality it imposes on other users, i.e. the welfare that others lose because of its presence. A downside of the VCG mechanism is that it may not generate any revenue. It is well-known that collusion has a negative effect on revenue in the VCG mechanism~\cite{AM06} --- while such collusion can even be illegal in other settings, in coopetitive ad auctions we specifically want advertisers to cooperate on ads that are of mutual benefit.

For example, two players can make their payments zero by simultaneously claim sufficiently large values for the winning outcome $o$. If the players' bids are sufficiently large that $o$ is still the welfare-maximizing outcome even if one of the pair were removed, the externality that each player imposes is zero, and nobody pays anything. This happens easily in coopetitive ad auctions. For example:
\begin{example}
In the following single slot ad auction $\{(\agentval{A}{1}, \agentval{B}{1},\agentval{C}{1},\agentval{D}{1}), (\agentval{E}{2.9})\}$, the VCG mechanism will show the first ad in the slot, but nobody pays anything.
\end{example}
In this case, the $ABCD$ ad remains the best ad even if a single bidder is removed. As a result, nobody pays anything. Such a scenario can occur naturally if the winning ad is valued by many small bidders.

The weak revenue of the VCG mechanism is not limited to extremes like the above example. In general,  payments will be lower than an auctioneer might hope:
\begin{example}
In the following single slot ad auction $\{(\agentval{A}{2}, \agentval{B}{2}), (\agentval{E}{3})\}$, the VCG mechanism will show the first ad in the slot and players $A$ and $B$ will each pay 1.
\end{example}
Intuition says that the the auctioneer should hope to make $A$ and $B$ pay a total of 3, since that is the total bid required to beat $E$. However, the total VCG payment for the first ad is only 2. In contrast, we show that a coopetitive first-price auction will indeed generate a revenue of 3.

\section{First-Price Auction Equilibria}
	In this section, we consider behavior of first-price auctions in the Coopetitive Ad Auction problem. Coopetitive auctions differ from the standard setting in that advertisers
	involved in the same ad are expected to cooperate on that ad; but remain competitors across other ads. In particular, if two losing advertisers can jointly raise their bids and win, we expect this to happen. 
	
	We thus restrict our attention to equilibria that are \emph{cooperatively envy-free}. By this, we mean that no cooperating advertisers in a losing ad can jointly raise
	their bids to beat out the presently winning ad. If a losing ad shares bidders with the winning ad, then it is only the advertisers not in either ad which affect this condition.

	\begin{definition} The bids $(b_i)_{i\in T}$ for the winning ad $T$ are cooperatively envy-free (CEF) if and only if for all alternate ads $S_j$, 
		\begin{eqnarray}
			\sum_{i\in T\setminus S_j} b_i \geq \sum_{i \in S_j\setminus T} v_i \
		\end{eqnarray}
	\end{definition}

	We also insist that the agents bidding are individually rational (IR) and bids are non-negative, that for each advertiser $i$, $0\leq \bid{i} \leq \val{i}$. 
	
	Combining IR and CEF yields efficiency - if the bids of agents in $T\setminus S_j$ are at least the values of agents in $S_j\setminus T$, then so too are the values.
	
	\begin{lemma}  If the bids $(b_i)_{i\in T}$ for the winning ad $T$ are IR and CEF, then $T$ is the efficient winning ad.
	\end{lemma}			

	These CEF, IR and non-negativity conditions form a polytope of possible payments associated with the correct winning ad. Not all of these are equilibria - those will
	instead form the Pareto frontier of the polytope. On this frontier, there can still be many possible equilibria. 
	Consider the following ad auction: $\{(\agentval{A}{100}, \agentval{B}{100}), (\agentval{C}{99})\}$,
	Every set of bids $b_A = x, b_B = 99-x$ for $0\leq x \leq 99$ constitute an equilibrium.
	
	Now, consider the equilibrium conditions. By CEF, no losing advertisers will be able to raise their bids and affect the outcome, 
	hence we need worry only about the winning bidders lowering their bids. 

	\begin{lemma}\label{equilibriumotherset}
		The IR, CEF, non-negative bids $(\bid{i})$ form an equilibrium for the winning ad $T$ if and only if for each bidder 
		$k$, $\bid{k}=0$ or there exists an ad $S_j\not\ni k$ s.t. 
		\begin{equation}
			\sum_{i\in T} \bid{i} = \sum_{i\in S_j} \bid{i}
		\end{equation}
	\end{lemma}
	\begin{proof}
		First, the `if' direction. Assume such an $S_j$ exists for every $k$. Then, were $k$ to lower his bid, $S_j$ would 
		win and $k$ would no longer be in the winning set.
		
		Consider the other direction. Assume a set of bids equilibrium bids are CEF, IR and non-negative. For every winning advertiser $k$, they must not be able to lower 
		their bids and still win - otherwise they would, and we would not be in equilibrium. Thus, for any $k$ s.t. $\bid{k}>0$, 
		there must be such a set $S_j$. 
	\end{proof}
			
	\subsection{Revenue}
		In this section, we consider the revenue behavior of equilibrium points in the polytope. First, note that the revenue is not the same for all equilibrium points - 
		this is not simply a matter of dividing a fixed payment up. As an example, consider the following three ad, five interested party setting: $\{(\agentval{A}{1}, \agentval{B}{1}, \agentval{C}{1}), (\agentval{A}{1}, \agentval{D}{1}), (\agentval{B}{1}, \agentval{E}{1})\}$. Clearly the first ad should win, but what should the payments be? 
	
	Our CEF, IR and non-negativity conditions give us the following polytope: $\bid{A} + \bid{C} \geq 1$, $\bid{B} + \bid{C} \geq 1$, $0 \leq \bid{A},\bid{B}, \bid{C} \leq 1$. 
	The set of equilibrium points includes $(1,0,1)$, $(0, 1, 0)$ and every convex combination of the two. Thus, the revenue of the equiliubrium points can range from $1$ to $2$.
	
	How would other mechanisms do? VCG will charge nothing, as no advertiser is integral to the ad being displayed. Were A and B to lie and say they are not affiliated with $C$ or $D$, 
	any second or first price mechanism would insist on a payment of $1$ from the three of them. In this example then, the revenue of our first price equilibria are lower
	bounded by VCG, and by simpler first and second price auctions.

	\begin{lemma}
	For every point in the polytope, the bid of each advertiser is at least their VCG payment.
	\end{lemma}	
	\begin{proof}
	Consider advertiser $i$. Let $T$ be the winning ad, and let $S_j$ be the winning ad without $i$. If $S_j=T$, then $i$'s VCG payment is $0$, and hence we need only worry about the case that $S_j\neq T$. By our CEF constraints, we have $b_i \geq \sum_{k\in S_j\setminus T} v_k - \sum_{k\in S_j\setminus T- \{i\}} b_k \geq \sum_{k\in S_j\setminus T} v_k - \sum_{k\in S_j\setminus T- \{i\}} v_k$. The latter quantity is exactly $i$'s VCG payment, and hence every advertiser's bid is at least their VCG payment.
	\end{proof}
	
	One possible way of dealing with this is to have all members of the winning ad pretend to be
	uninterested in the losing ads, and then have them split the payment of the second highest bidding ad. Any equilibrium in the polytope will get at least the revenue of such a mechanism:

	\begin{lemma} \label{ignoringlosers}
	The revenue of any equilibrium in the polytope is at least the maximum total value of non-winning advertisers in a non-winning ad.
	\end{lemma}	
	
	This follows directly from the CEF conditions. Note that this is the same as the revenue a second price auction would get if treated advertisements as single agents and removed interests of winning advertisers in losing ads.

	Thus, any first price equilibria that satisfies CEF (eg., no losing advertisers can collaborate to increase their bids and win) has good revenue - revenue that beats both VCG and a natural analogue to a second price auction.

\section{The Egalitarian Solution}
	As discussed earlier, there are many ways the winners can split payments while still satisfying our equilibrium 
	and cooperative envy-freeness constraints. In a first price auction, the exact split that we expect bidders to reach 
	will depend on the preferences and bidding dynamics of the advertisers.
	
	In this section, we consider one such split in particular --- the \emph{egalitarian bargaining solution}.
	In the egalitarian solution, the utility of the worst-off player is maximized, and on up the line. This can be defined as the equilibrium with the lexicographically maximum surplus:

	\begin{definition}
		The egalitarian solution in the coopetitive first-price auction is an equilibrium that displays the highest surplus ad and charges advertisers so that the surplus vector is lexicographically maximal when bidders are ordered in terms of increasing surplus.
	\end{definition}
	In many normal settings, this will result in all players sharing equally the surplus generated.

	\subsection{Algorithm}
	The egalitarian equilibrium can be computed by repeatedly lowering bids uniformly. The algorithm is described in Algorithm~\ref{alg:egal}.

\begin{algorithm}[tb]\label{alg:egal}
\SetKwInOut{Input}{input}\SetKwInOut{Output}{output}
\SetKwIF{WP}{ElseWP}{Otherwise}{with probability}{}{with probability}{otherwise}{end}
\SetKw{KwSet}{Set}
\SetKwRepeat{DoUntil}{do}{until}
\Input{A coopetitive ad auction problem.}
\Output{The egalitarian equilibrium bids $b_i$.}
\DontPrintSemicolon
\BlankLine
\nl Set the bids of all advertisers to their values. Call $T$ the winning ad.\;
\nl Lower bids of advertisers in the winning ad $T$ uniformly until some bidder $i$ reaches $b_i=0$ or a constraint $\sum_{i\in T} b_i \geq \sum_{i\in S_j} b_i$ would be violated for some ad $S_j$.\;
\nl Fix the bids of advertisers in $T\setminus S_j$ (or fix the bid of $i$ if $b_i$ reached 0).\;
\nl Repeat (2) and (3), lowering only unfixed bids until all bidders in $T$ are fixed.\;
caption{An algorithm for computing the egalitarian equilibrium in the single-slot first-price auction.}
\end{algorithm}
					
	\begin{lemma}
		Algorithm~\ref{alg:egal} computes the egalitarian equilibrium point.	
	\end{lemma}
	\begin{proof}
	
		First, we show that the resulting point is efficient and cooperatively envy-free --- then we'll show that the resulting equilibrium must be the egalitarian one.

		Begin with the following claim:
		\begin{claim}
			The total bid for the winning ad $T$ never drops below the total bids of non-winning ads.
		\end{claim}
		At any point, the algorithm lowers only the bids of players in \emph{every} ad tied for the highest 
		value. As a result, every ad tied for the highest value is decreasing by the same amount. Thus, the 
		winning ad at the beginning of the algorithm remains the winning ad at the end and hence the final winning
		ad is the ad with the most surplus. 
		
		At the end of the auction, as $T$ remains the ad with the highest value and no non-winning advertisers
		see their surpluses decrease, the CEF constraints will be satisfied for every alternate ad $S_j$.
	
		 By Lemma \ref{equilibriumotherset} we know that our point will be in equilibrium if and only if 
		 there is a set for every agent tied with the winning set that he is not in. Our algorithm will only 
		 stop when for every bidder there is such a set, or they are bidding $0$ --- hence it must be such an equilibrium.

	Thus, we've argued that the final point is an efficient, CEF point. Now, we discuss whether or not it is in fact the 
	egalitarian solution. We'll prove this with induction. Assume that the algorithm gives the egalitarian bargaining solution for the
	first $i-1$ lowest surplus advertisers. Consider the algorithm after those advertisers are fixed - in particular, the next time 
	an advertiser has their bid fixed, with a surplus of $z$. At this point, there will be a set $S_j$ s.t.
	\begin{equation}
		\sum_{i\in T\setminus S_j} \max (v_i-z, b'_i) = \sum_{i\in T\setminus S_j} v_i
	\end{equation}

	First, note that the algorithm cannot reduce the bid of $i$ further than the egalitarian solution --- otherwise that would be
	the egalitarian solution. Assume now that the algorithm results in a lower surplus for player $i$ --- hence, $i$'s bid is fixed 
	before being lowered to his egalitarian surplus. By our assumption, then, 
	$\sum_{i\in T\setminus S_j} \max (v_i-z, b'_i) > \sum_{i\in T\setminus S_j} b'_i \geq \sum_{i\in T\setminus S_j} b_i$. By our CEF 
	constraints in the egalitarian solution, we have that $\sum_{i\in T\setminus S_j} b_i \geq  \sum_{i\in T\setminus S_j} v_i$. Then 
	$\sum_{i\in T\setminus S_j} \max (v_i-z, b'_i) < \sum_{i\in T\setminus S_j} b_i$ and hence $\sum_{i\in T\setminus S_j} b'_i = \sum_{i\in T\setminus S_j} b_i$. 
	Since all advertisers with surplus less than $z$ have the egalitarian surplus, and all advertisers with more surplus are in $S_j$, $b'_i=b_i$.
	
	\end{proof}

\section{Conclusion and Open Problems}

As we have discussed, a wide variety of ads provide value to more than one party --- ads for computers, cell phones, and even baseball teams to name a few. As evidenced by the  ``samsung intel laptop'' query shown in Figure~\ref{fig:silsearch}, current ad auctions completely fail to account for such shared benefits. While Google may enjoy a little extra revenue at Intel's expense in the interim, Figure~\ref{fig:silsearch} does not represent a sustainable equilibrium --- our theoretical results show that failure to account for the shared benefit of an advertisement can have a substantial negative effect on both welfare and revenue.

As a possible solution, we showed that the coopetitive first-price auction behaves well in equilibrium. In particular, all equilibria that satisfy cooperative envy-freeness are efficient and have good revenue. Moreover, such equilibria can be computed efficiently. Finally, as noted in Section~\ref{sec:prelim}, these results generalize to a model where bidders' have different values for different ads.

The main open question is {\em what is the best way for auctioneers to accommodate the shared value of an ad?} While the first-price auction may have nice equilibria and offer transparency to bidders, it is not clear that it is the best solution. A big downside of a first-price auction is that bidders have an incentive to tweak their bids to lower their payments. In contrast, a small modification to a bid in a GSP mechanism will not affect payments unless it affects the outcome of the auction. Unfortunately, na\"{i}vely building a GSP-style auction will either inherit the downsides of VCG or allow players to affect the distribution of payments by tweaking their bids. First-price auctions also have problems in multi-slot settings, as evidenced by the poor performance of Overture's original first-price auction. Thus, while the properties of first-price equilibria presented herein are desirable, the question of how auctioneers should design real auctions remains unresolved.

\bibliographystyle{plain}
\bibliography{miaa}

\begin{thebibliography}{10}

\bibitem{AM06}
Lawrence~M. Ausubel and Paul Milgrom.
\newblock The lovely but lonely vickrey auction.
\newblock In {\em Combinatorial Auctions, chapter 1}. MIT Press, 2006.

\bibitem{B72}
Paul~D. Berger.
\newblock Vertical cooperative advertising ventures.
\newblock {\em Journal of Marketing Research}, 9(3):pp. 309--312, 1972.

\bibitem{EOS07}
Benjamin Edelman, Michael Ostrovsky, and Michael Schwarz.
\newblock Internet advertising and the generalized second-price auction:
  Selling billions of dollars worth of keywords.
\newblock {\em American Economic Review}, 97(1):242--259, March 2007.

\bibitem{HPS09}
Xiuli He, Ashutosh Prasad, and Suresh~P. Sethi.
\newblock Cooperative advertising and pricing in a dynamic stochastic supply
  chain: Feedback stackelberg strategies.
\newblock {\em Production and Operations Management}, 18(1):78--94, 2009.

\bibitem{HLM02}
Zhimin Huang, Susan~X. Li, and Vijay Mahajan.
\newblock An analysis of manufacturer-retailer supply chain coordination in
  cooperative advertising*.
\newblock {\em Decision Sciences}, 33(3):469--494, 2002.

\bibitem{II}
{Intel Corporation}.
\newblock Intel inside program - anatomy of a brand campaign.
\newblock \url{http://www.intel.com/pressroom/intel\_inside.htm}.

\bibitem{JMS99}
Philippe Jehiel, Benny Moldovanu, and Ennio Stacchetti.
\newblock Multidimensional mechanism design for auctions with externalities.
\newblock {\em Journal of Economic Theory}, 85(2):258 -- 293, 1999.

\bibitem{K98}
Raja Kali.
\newblock Minimum advertised price.
\newblock {\em Journal of Economics \& Management Strategy}, 7(4):647--668,
  1998.

\bibitem{KMSW10}
Piotr Krysta, Tomasz Michalak, Tuomas Sandholm, and Michael Wooldridge.
\newblock Combinatorial auctions with externalities.
\newblock In {\em Proceedings of the 9th International Conference on Autonomous
  Agents and Multiagent Systems: volume 1 - Volume 1}, AAMAS '10, pages
  1471--1472, Richland, SC, 2010. International Foundation for Autonomous
  Agents and Multiagent Systems.

\bibitem{N06}
Matthew Nagler.
\newblock An exploratory analysis of the determinants of cooperative
  advertising participation rates.
\newblock {\em Marketing Letters}, 17(2):91--102, April 2006.

\bibitem{W}
{Wikipedia}.
\newblock Coopetition.
\newblock \url{http://en.wikipedia.org/wiki/Coopetition}.

\end{thebibliography}

\end{document}